\documentclass{article}
\usepackage{amsmath}
\usepackage{amsthm}
\usepackage{amssymb}
\usepackage{graphicx}
\usepackage{hyperref}
\usepackage{subfig}
\usepackage{booktabs}
\usepackage{array}
\usepackage{natbib}

\newtheorem{thm}{Theorem}[section]
\newtheorem{lem}{Lemma}[section]
\newtheorem{cor}{Corollary}[section]
\newtheorem{prop}{Proposition}[section]
\theoremstyle{definition}
\newtheorem{defn}{Definition}[section]
\newtheorem{remark}{Remark}

\newcolumntype{C}[1]{>{\centering\arraybackslash}p{#1}}

\begin{document}

\title{Bounds on Longest Simple Cycles in Weighted Directed Graphs via
  Optimum Cycle Means\thanks{The core ideas of this work were
    developed during the author's graduate research in the late
    1990s.}}

\author{Ali Dasdan\thanks{This article was revised with AI assistance (Claude, Codex, Gemini). All scientific content and conclusions are the author's own, as is responsibility for any errors.}\\
  KD Consulting\\
  Saratoga, CA, USA\\
  alidasdan@gmail.com\\
}

\maketitle

\begin{abstract}
The problem of finding the longest simple cycle in a directed graph is
NP-hard, with critical applications in computational biology,
scheduling, and network analysis. Existing approaches include exact
algorithms with exponential runtimes, approximation algorithms limited
to specific graph classes, and heuristics with no formal
guarantees. In this paper, we exploit \textit{optimum cycle means}
(minimum and maximum cycle means), computable in strongly polynomial
time, to derive both \textit{strict bounds} and \textit{heuristic
  estimates} for the weight and length of the longest simple cycle in
general graphs. The strict bounds can prune search spaces in exact
algorithms while the heuristic estimates (the arithmetic mean
$\lambda_{avg}$ and geometric mean $\lambda_{geo}$ of the optimum
cycle means) guarantee bounded approximation error. Crucially, a
single computation of optimum cycle means yields both the bounds and
the heuristic estimates. Experimental evaluation on ISCAS benchmark
circuits yields two findings. The strict lower bounds on the weight
and length of the longest cycle are loose, with medians 57--95\% below
the true values. The heuristic estimates of the longest cycle's mean,
$\lambda_{avg}$ and $\lambda_{geo}$, have median relative errors of
6--13\% vs. 11--21\% for symmetric (uniform) weights and 41--92\%
vs. 25--35\% for skewed (log-normal) weights, favoring $\lambda_{avg}$
for symmetric distributions and $\lambda_{geo}$ for skewed
distributions.
\end{abstract}

\section{Introduction}
\label{sec:intro}

Finding the longest cycle in a graph (with $n$ nodes and $m$ edges) is
a fundamental combinatorial optimization problem with applications in
circuit design, scheduling, and network analysis. While the shortest
cycle can be found in polynomial time in graphs without negative
cycles \citep{OrSe2017}, finding the
longest simple cycle (maximum length or maximum weight) is NP-hard
\citep{Ga79}. Moreover, the problem resists efficient approximation:
it cannot be approximated within a factor of $n^{1-\epsilon}$ for any
$\epsilon > 0$ unless P=NP \citep{BjHuKh04}. Given this hardness,
practitioners rely on heuristics, branch-and-bound algorithms
\citep{MoJaSaSe16,StPuFe14,CoStFe20}, or exhaustive enumeration on
small graphs using exponential-time algorithms
\citep{Jo75,HaJa08,BaSa25}. Polynomial-time computable bounds on the
longest cycle can aid all these approaches by pruning search spaces,
verifying feasibility, or providing quick estimates.

This paper investigates the utility of \textit{Optimum Cycle Means},
specifically the minimum ($\lambda_{min}$) and maximum
($\lambda_{max}$) cycle means, as a basis for such bounds. \textit{The
  key insight is that the cycle mean of any cycle, including the
  longest, must lie between the minimum and maximum cycle
  means}. Since these optimum cycle means can be computed efficiently
in theory (strongly polynomial $O(nm)$ time) and in practice
\citep{Da04}, they offer an attractive trade-off between computational
cost and bounding power. Crucially, a single computation yields both
the strict bounds and the heuristic estimates.

However, effectively leveraging this framework requires distinguishing
between the \textit{strict bounds} derived from critical cycles and
the \textit{heuristic estimates} based on cycle means. To evaluate the
robustness of these estimates, we consider multiple edge weight
distributions. Specifically, in addition to uniform weights, we
evaluate our bounds on graphs with log-normally distributed
weights. This is motivated by applications in digital circuit timing
analysis, where log-normal distributions have been used to model
interconnect delays \citep{AlLiKa03} and the variability of cell
delays under process variations \citep{KiHaPa16}.

Our investigation yields two distinct findings:
\begin{enumerate}
\item \textbf{Strict bounds are guaranteed but loose:}
  We derive algebraic bounds linking longest cycles to ``critical
  cycles'' (those achieving optimum cycle means). While theoretically
  sound, critical cycles tend to be topologically short, making the 
  resulting lower bounds often loose (median 57--95\% below true values).
\item \textbf{Heuristic estimates are tighter with bounded error:} We
  propose two heuristic estimates: the arithmetic mean
  $\lambda_{avg}$ and the geometric mean $\lambda_{geo}$ of the
  optimum cycle means. These guarantee bounded approximation error, 
  with accuracy depending on distribution skewness: $\lambda_{avg}$ 
  is preferable for symmetric (uniform) distributions, while
  $\lambda_{geo}$ is preferable for skewed (log-normal) distributions.
\end{enumerate}

Our specific contributions are as follows:
\begin{itemize}
    \item We derive formal algebraic bounds on the weight, length, and
      mean of max-weight and max-length cycles based on optimum cycle
      means (\S~\ref{sec:graph}--\S~\ref{sec:extensions});
    \item We establish dual results for shortest cycles
      (\S~\ref{sec:shortest-results});
    \item We propose heuristic estimates and bound tightness
      analysis (\S~\ref{sec:heuristic});
    \item We provide an extensive experimental evaluation on ISCAS
      benchmark graphs with both uniform and log-normal edge weights,
      analyzing how distribution properties affect bound quality
      (\S~\ref{sec:benchmark-results});
    \item We discuss the limitations and applications of the proposed
      approach (\S~\ref{sec:limits} and \S~\ref{sec:applications}).
\end{itemize}

Related work is discussed in \S~\ref{sec:related-work}, and we
summarize our conclusions and avenues for future research in
\S~\ref{sec:conclusions}.

\section{Graph Basics}
\label{sec:graph}

Let $G = (V, E, w)$ be a finite directed graph with $|V| = n$ nodes
and $|E| = m$ directed edges (arcs). The weight function $w: E \to
\mathbb{R}$ maps edges to real values (negative, zero, or
positive). We denote the weight of an edge $e=(u,v)$ as $w(e)$ or
$w(u,v)$. The minimum and maximum individual edge weights are denoted
by $w_{min}$ and $w_{max}$, respectively.

Without loss of generality, we assume $G$ is strongly connected and
contains at least one cycle. If $G$ is not strongly connected, the
analysis can be applied independently to each Strongly Connected
Component (SCC), which can be identified in $O(n+m)$ time
\citep{MeSa08}.

The weight $w(C)$ of a cycle $C$ is the sum of its edge weights, and
its length $|C|$ is the number of edges it contains. Throughout,
cycles are simple, that is, they repeat no node, so $|C| \leq n$;
self-loops are allowed. A cycle is termed positive, zero, or negative
based on the sign of $w(C)$. A cycle of length 1 is a self-loop.

\section{Optimum Cycle Means}
\label{sec:opt-means}

\begin{defn}
  \label{def:means}
The \textit{cycle mean} $\lambda(C)$ of a cycle $C$ is its average
edge weight:
\begin{equation}
  \label{eq:means}
  \lambda(C) = \frac{\sum\limits_{e\in C} w(e)}{|C|} = \frac{w(C)}{|C|},
\end{equation}
where $|C| \geq 1$.
\end{defn}

\begin{defn}
  \label{def:opt-means}
The \textit{maximum cycle mean} $\lambda_{max}$ and \textit{minimum
  cycle mean} $\lambda_{min}$ of $G$ are the optimum values of cycle
means over all cycles in $G$. Formally,
\begin{equation}
  \label{eq:opt-means}
  \lambda_{max} = \max\limits_{C\in G} \lambda(C) \quad \text{and}
  \quad \lambda_{min} = \min\limits_{C\in G} \lambda(C).
\end{equation}
\end{defn}
Both are computable efficiently in $O(mn)$ time \citep{Da04}, without
needing to enumerate over all cycles in $G$ explicitly.

Optimum cycle means admit a linear programming formulation. For
example, $\lambda_{min}$ is the maximum $\lambda$ satisfying:
\begin{equation}
  d(v) - d(u) \leq w(u,v) - \lambda, \quad \forall (u,v) \in E,
\end{equation}
where $d(v)$ represents a potential function on nodes. At optimality,
$\lambda=\lambda_{min}$, $d(v)$ can be chosen as the shortest path
distance from an arbitrary source node under the modified weights
$w'(u,v)=w(u,v) - \lambda_{min}$, and $G$ contains no negative cycles
under these modified weights.

\section{Criticality}
\label{sec:criticality}

An edge $(u,v)$ is termed \textit{min-critical} if it lies on a
critical path in the linear program above, satisfying
$d(v)-d(u)=w(u,v)-\lambda_{min}$. A cycle $C$ is min-critical if
$\lambda(C)=\lambda_{min}$; we denote such a cycle by $C_{min}$. The
subgraph $G_{min}$ containing all min-critical edges is the
min-critical subgraph. Analogous definitions apply for
$\lambda_{max}$, max-critical cycles ($C_{max}$), and the max-critical
subgraph.

Multiple critical cycles may exist with the same cycle mean but
different lengths. Throughout, $\mathcal{C}_{min}$ and
$\mathcal{C}_{max}$ denote the nonempty sets of min-critical and
max-critical cycles that have actually been found; they may be proper
subsets of all critical cycles, and every bound below holds for any
such subsets. Since our bounds involve $|C_{min}|$ and $|C_{max}|$, we
optimize over the critical cycles available:
\begin{itemize}
  \item For \textit{lower bounds}, take $\max$ over the available
    critical cycles to obtain the tightest (largest) lower bound.
  \item For \textit{upper bounds}, take $\min$ over the available
    critical cycles to obtain the tightest (smallest) upper bound.
\end{itemize}
Optimum cycle mean algorithms naturally produce at least one critical
cycle as a witness. Every cycle of the critical subgraph is critical,
because its edges are all tight and the reweighted edge weights
$w(u,v) - \lambda_{min}$ telescope to zero around the cycle. Hence a
shortest critical cycle can be found in polynomial time by a
shortest-cycle search in the critical subgraph, whereas finding a
longest critical cycle is NP-hard. Running multiple algorithms or
enumerating cycles of the critical subgraph may yield critical cycles
of varying lengths, enabling tighter bounds.

\section{Cycle Mean Inequalities}
\label{sec:basic-results}

The definition of optimum cycle means implies the following
fundamental bounds for any cycle $C$:
\begin{lem}[Basic cycle mean bounds]
\label{lem:basic-bounds}
For any cycle $C$ in $G$:
\begin{equation}
  \label{eq:anycycle1}
  w_{min} \leq \lambda_{min} \leq \frac{w(C)}{|C|} \leq \lambda_{max} \leq w_{max},
\end{equation}
or equivalently,
\begin{equation}
  \label{eq:anycycle2}
  |C|w_{min} \leq |C|\lambda_{min} \leq w(C) \leq |C|\lambda_{max} \leq |C|w_{max}.
\end{equation}
\end{lem}
\begin{proof}
The inequalities $\lambda_{min} \leq \lambda(C) \leq \lambda_{max}$
follow by definition. The outer bounds hold because no cycle mean can
exceed the maximum individual edge weight or be less than the minimum
individual edge weight.
\end{proof}

As mentioned in \S~\ref{sec:graph}, we assume $G$ is strongly
connected, meaning it has at least one cycle. Hence, it is also
guaranteed that $G$ contains at least one min-critical cycle $C_{min}$
and one max-critical cycle $C_{max}$, such that $\lambda_{min} =
\lambda(C_{min})$ and $\lambda_{max} = \lambda(C_{max})$.

\section{Longest Cycle Bounds: Basics}
\label{sec:longest-basics}

We distinguish two notions of ``longest'' cycle: a \textit{max-weight
cycle} maximizes total weight $w(C)$, while a \textit{max-length cycle}
maximizes edge count $|C|$.

Let $L$ denote a max-weight or a max-length cycle. Trivial bounds on
$w(L)$ can be derived from edge weights alone:
\begin{equation}
  |L|w_{min} \leq w(L) \leq |L|w_{max}.
\end{equation}
However, optimum cycle means provide tighter constraints:
\begin{lem}[Cycle mean bounds for longest cycles]
\label{lem:longest-mean-bounds}
Let $L$ be either a max-weight or max-length cycle in $G$. Then:
\begin{equation}
  \label{eq:longest1}
  \lambda_{min} \leq \lambda(L) \leq \lambda_{max},
\end{equation}
which implies:
\begin{equation}
  \label{eq:longest2}
  |L|\lambda_{min} \leq w(L) \leq |L|\lambda_{max}.
\end{equation}
\end{lem}
\begin{proof}
This is a direct application of Lemma~\ref{lem:basic-bounds} to the
specific cycle $L$.
\end{proof}

While the exact value of $|L|$ is unknown,
Lemma~\ref{lem:longest-mean-bounds} establishes necessary conditions
for the weight of $L$:
\begin{cor}[Sign conditions]
  \label{cor:sign-conditions}
\mbox{}\par
\begin{enumerate}
  \item $\lambda_{min} > 0 \Rightarrow w(L) > 0$.
  \item $\lambda_{max} < 0 \Rightarrow w(L) < 0$.
  \item $\lambda_{min} = 0 \Rightarrow w(L) \geq 0$.
  \item $\lambda_{max} = 0 \Rightarrow w(L_w) = 0$ and $w(L_l) \leq 0$.
\end{enumerate}
\end{cor}

\begin{proof}
Items 1--3 follow from $|L|\lambda_{min} \leq w(L) \leq
|L|\lambda_{max}$. For item 4, $\lambda_{max} = 0$ means every cycle
has non-positive weight, so $w(L_l) \leq 0$; since a zero-mean cycle
exists and $L_w$ maximizes weight, $w(L_w) = 0$. The inequality for
$L_l$ cannot be strengthened to equality: a zero-weight self-loop next
to a two-edge cycle of negative weight gives $\lambda_{max} = 0$ but
$w(L_l) < 0$.
\end{proof}

\section{Longest Cycle Bounds: Extensions}
\label{sec:extensions}

We can refine these bounds by analyzing the algebraic properties of
max-weight ($L_w$) and max-length ($L_l$) cycles separately.
\begin{lem}
\label{lem:mean_domination}
If $\lambda_{max} \geq 0$, the cycle mean of the max-weight cycle
dominates that of the max-length cycle:
\begin{equation}
  \lambda(L_{w}) \geq \lambda(L_{l}).
\end{equation}
\end{lem}
\begin{proof}
If $\lambda_{max} = 0$, all cycles have non-positive mean, so $w(L_w)
\leq |L_w| \lambda_{max} = 0$. But since a zero-mean cycle exists,
$w(L_w) \geq 0$, hence $w(L_w) = 0$ and $\lambda(L_w) = 0 \geq
\lambda(L_l)$.

If $\lambda_{max} > 0$, some cycle has positive weight, so $w(L_w) >
0$ and $\lambda(L_w) > 0$. If $w(L_l) \leq 0$, then $\lambda(L_l) \leq
0 < \lambda(L_w)$. If $w(L_l) > 0$, then since $w(L_w) \geq w(L_l)$
and $|L_w| \leq |L_l|$:
$$ \lambda(L_{w}) = \frac{w(L_{w})}{|L_{w}|} \geq
\frac{w(L_{l})}{|L_{w}|} \geq \frac{w(L_{l})}{|L_{l}|} =
\lambda(L_{l}). $$
\end{proof}

The condition $\lambda_{max} \geq 0$ is necessary. When $\lambda_{max}
< 0$, the relationship between $\lambda(L_w)$ and $\lambda(L_l)$ is
indeterminate; the mean of the max-length cycle may exceed that of the
max-weight cycle depending on the distribution of edge weights.

Optimum cycle mean algorithms return critical cycles as witnesses, and
any witness cycle is itself a feasible solution to the longest cycle
problem. This observation yields bounds that need no sign conditions
and that the results below build upon.

\begin{lem}[Witness bounds]
\label{lem:witness}
Let $\mathcal{W}$ be any nonempty set of cycles of $G$ (in practice, the
critical cycles found), and let $B_w = \max\limits_{C \in \mathcal{W}} w(C)$
and $B_l = \max\limits_{C \in \mathcal{W}} |C|$. Then:
\begin{enumerate}
\item $w(L_w) \geq B_w$ and $|L_l| \geq B_l$.
\item If $\lambda_{max} > 0$ and $B_w \geq 0$, then
  \begin{equation*}
    |L_w| \geq \lceil B_w/\lambda_{max} \rceil
    \quad \text{and} \quad
    \lambda(L_w) \geq \max\{\lambda_{min}, B_w/n\}.
  \end{equation*}
\item If $\lambda_{min} \geq 0$, then $w(L_l) \geq B_l\lambda_{min}$.
\end{enumerate}
\end{lem}

\begin{proof}
Item 1 holds because $L_w$ maximizes weight and $L_l$ maximizes length
over all cycles, including those in $\mathcal{W}$. For item 2,
Lemma~\ref{lem:longest-mean-bounds} gives $B_w \leq w(L_w) \leq
|L_w|\lambda_{max}$; dividing by $\lambda_{max} > 0$ and using the
integrality of $|L_w|$ gives the length bound, and $\lambda(L_w) =
w(L_w)/|L_w| \geq B_w/|L_w| \geq B_w/n$ since $B_w \geq 0$ and $|L_w| \leq
n$. For item 3, $w(L_l) \geq |L_l|\lambda_{min} \geq B_l\lambda_{min}$
since $\lambda_{min} \geq 0$.
\end{proof}

In what follows, $\mathcal{W} = \mathcal{C}_{min} \cup
\mathcal{C}_{max}$ is the set of critical cycles available, so that
$B_w = \max\{\max\limits_{C \in \mathcal{C}_{min}} w(C),
\max\limits_{C \in \mathcal{C}_{max}} w(C)\}$ and $B_l = \max\{\max\limits_{C \in
  \mathcal{C}_{min}} |C|, \max\limits_{C \in \mathcal{C}_{max}} |C|\}$, where
$\mathcal{C}_{min}$ and $\mathcal{C}_{max}$ are the sets of critical
cycles found (Section~\ref{sec:criticality}). We refer to these as
the \textit{critical witness bounds} $B_w$ and $B_l$. Since $\mathcal{W}
\supseteq \mathcal{C}_{max}$, $B_w \geq w(C_{max}) =
|C_{max}|\lambda_{max}$ for every $C_{max} \in \mathcal{C}_{max}$, so
$B_w > 0$ whenever $\lambda_{max} > 0$; likewise $B_l \geq |C|$ for every
critical cycle $C$ found.

\subsection{Max-weight Cycles}
\label{sec:max-weight-case}

Let $L_w$ be a max-weight cycle and $\mathcal{C}_{max}$ be the set of
max-critical cycles $C_{max}$ found in $G$. By definition, $w(L_w) \geq
w(C)$ for all cycles $C$ in $G$.

\begin{lem}[Max-weight key inequality]
\label{lem:max-weight-key}
$\lambda_{max}(|L_w| - |C_{max}|) \geq 0$.
\end{lem}

\begin{proof}
Since $L_w$ maximizes weight, $w(L_w) \geq w(C_{max}) =
|C_{max}|\lambda_{max}$ for any $C_{max} \in \mathcal{C}_{max}$.  From
Lemma~\ref{lem:longest-mean-bounds}, we have $w(L_w) \leq
|L_w|\lambda_{max}$.  Combining these yields:
$$ |C_{max}|\lambda_{max} \leq w(L_w) \leq |L_w|\lambda_{max} \implies \lambda_{max}(|L_w| - |C_{max}|) \geq 0. $$
\end{proof}

\begin{thm}[Max-weight cycle bounds]
  \label{thm:max-weight}
  Let $B_w$ be the critical witness bound, so that $\mathcal{W}
  \supseteq \mathcal{C}_{max}$. We derive the following three cases
  from Lemmas~\ref{lem:witness} and~\ref{lem:max-weight-key}:
  \begin{enumerate}
  \item If $\lambda_{max} > 0$, then
    $|L_w| \geq \lceil B_w/\lambda_{max} \rceil \geq \max\limits_{C \in \mathcal{C}_{max}} |C|$
    and
    $w(L_w) \geq B_w \geq (\max\limits_{C \in \mathcal{C}_{max}} |C|) \lambda_{min}$.
  \item If $\lambda_{max} < 0$, then
    $|L_w| \leq \min\limits_{C \in \mathcal{C}_{max}} |C|$
    and
    $(\min\limits_{C \in \mathcal{C}_{max}} |C|)\lambda_{max} \leq w(L_w) \leq \lambda_{max}$.
    \item If $\lambda_{max} = 0$, then $w(L_w) = 0$.
\end{enumerate}
\end{thm}

\begin{proof}
\mbox{}\par
\begin{itemize}
    \item \textbf{Case 1 ($\lambda_{max} > 0$):} Since $B_w > 0$,
      Lemma~\ref{lem:witness} gives $w(L_w) \geq B_w$ and $|L_w| \geq
      \lceil B_w/\lambda_{max} \rceil$. For any $C_{max}$, $B_w \geq
      w(C_{max}) = |C_{max}|\lambda_{max} \geq |C_{max}|\lambda_{min}$,
      so $\lceil B_w/\lambda_{max} \rceil \geq |C_{max}|$ and $B_w \geq
      |C_{max}|\lambda_{min}$. Since this holds for \textit{any}
      $C_{max}$, it holds for the maximum length in the set. (The
      weaker length bound also follows directly from
      Lemma~\ref{lem:max-weight-key} by dividing by $\lambda_{max}$.)
    \item \textbf{Case 2 ($\lambda_{max} < 0$):} Dividing
      Lemma~\ref{lem:max-weight-key} by negative $\lambda_{max}$
      reverses the inequality: $|L_w| \leq |C_{max}|$. Since this
      holds for \textit{any} $C_{max}$, it must hold for the minimum
      length in the set. For the weight, $w(L_w) \geq w(C_{max}) =
      |C_{max}|\lambda_{max}$ for every $C_{max}$; as $\lambda_{max} <
      0$, the largest of these values is attained by the shortest
      max-critical cycle, which gives the lower bound. The upper bound
      follows from $w(L_w) \leq |L_w|\lambda_{max} \leq \lambda_{max}$
      since $|L_w| \geq 1$ and $\lambda_{max} < 0$. Note that the
      length bound $|L_w| \leq \min |C|$ cannot be turned into a
      weight upper bound by multiplying with $\lambda_{max}$, as that
      reverses the inequality: with a self-loop of weight $-1.5$ and a
      two-edge cycle whose edges weigh $-1$ each, $\lambda_{max} =
      -1$, $\min |C_{max}| = 2$, and $w(L_w) = -1.5 > -2$.
    \item \textbf{Case 3 ($\lambda_{max} = 0$):} Then $w(C_{max})=0$,
      so $w(L_w) \ge 0$. Since $\lambda_{max}=0$ precludes positive
      cycles, $w(L_w)=0$.
\end{itemize}
\end{proof}

\subsection{Max-length Case}
\label{sec:max-length-case}

Let $L_l$ be a max-length cycle and $\mathcal{C}_{min}$ be the set of
min-critical cycles $C_{min}$ found in $G$. By definition, $|L_l| \geq
|C|$ for all cycles $C$ in $G$.

\begin{lem}[Max-length key inequality]
\label{lem:max-length-key}
$w(L_l) \geq w(C_{min})|L_l| / |C_{min}| = |L_l|\lambda_{min}$.
\end{lem}

\begin{proof}
From $\lambda(L_l) \geq \lambda_{min} = w(C_{min})/|C_{min}|$, multiply
both sides by $|L_l|$.
\end{proof}

\begin{thm}[Max-length cycle bounds]
  \label{thm:max-length}
  We derive the following three cases from
  Lemma~\ref{lem:max-length-key}:
  \begin{enumerate}
    \item If $\lambda_{min} > 0$, then $w(L_l) \geq B_l\lambda_{min}
      \geq \max\limits_{C \in \mathcal{C}_{min}} w(C)$, where $B_l$ is
      the critical witness bound, so that $\mathcal{W} \supseteq
      \mathcal{C}_{min}$.
    \item If $\lambda_{min} < 0$, then no direct comparison between
      $w(L_l)$ and $w(C_{min})$ is possible without knowing $|L_l|$,
      so we fall back to the unconditional bounds in
      Lemma~\ref{lem:longest-mean-bounds}.
    \item If $\lambda_{min} = 0$, $w(L_l) \geq 0$ as shown in
      Corollary~\ref{cor:sign-conditions}.
\end{enumerate}
\end{thm}

\begin{proof}
\mbox{}\par
\begin{itemize}
  \item \textbf{Case 1 ($\lambda_{min} > 0$):} By
    Lemma~\ref{lem:witness}, $|L_l| \geq B_l \geq |C_{min}|$ for every
    $C_{min}$. Since $\lambda_{min} > 0$, multiplying by
    $\lambda_{min}$ yields $|L_l|\lambda_{min} \geq B_l\lambda_{min}
    \geq |C_{min}|\lambda_{min} = w(C_{min})$. By
    Lemma~\ref{lem:max-length-key} and since $C_{min}$ is arbitrary,
    this implies $w(L_l) \geq B_l\lambda_{min} \geq \max\limits_{C \in
      \mathcal{C}_{min}} w(C)$.
  \item \textbf{Case 2 ($\lambda_{min} < 0$):} The inequality $|L_l|
    \geq |C_{min}|$ implies $|L_l|\lambda_{min} \leq
    |C_{min}|\lambda_{min}$. The lower bound weakens, so we rely on
    the general upper bound $w(L_l) \leq |L_l|\lambda_{max}$.
  \item \textbf{Case 3 ($\lambda_{min} = 0$):} Proven above.
\end{itemize}
\end{proof}

\begin{cor}[Length bound for max-length cycles]
  \label{cor:max-length-length-bound}
  Let $B_l$ be the critical witness bound. Without any sign condition,
  \begin{equation}
    |L_l| \geq B_l \geq \max\limits_{C \in \mathcal{C}_{min}} |C|.
  \end{equation}
  When $\lambda_{min} > 0$, this dominates the bound $|L_l| \geq
  \max\limits_{C \in \mathcal{C}_{min}} w(C)/\lambda_{max}$ obtained by
  combining Theorem~\ref{thm:max-length} with
  Lemma~\ref{lem:longest-mean-bounds}.
\end{cor}

\begin{proof}
  The first claim is item 1 of Lemma~\ref{lem:witness}. For the
  second, $w(C)/\lambda_{max} = |C|\lambda_{min}/\lambda_{max} \leq
  |C|$ for every $C \in \mathcal{C}_{min}$, since $0 < \lambda_{min}
  \leq \lambda_{max}$.
\end{proof}

\begin{remark}[Weight shifting]
  \label{rem:weight-shift}
  Subtracting $\lambda_{min}$ from all edge weights yields a
  transformed graph with $\lambda'_{min} = 0$ and $\lambda'_{max} =
  \lambda_{max} - \lambda_{min}$. While this transformation alters the
  max-weight cycle, it preserves the max-length cycle. Applying Case~3
  of Theorem~\ref{thm:max-length} to the transformed graph gives
  $w'(L_l) \geq 0$, which converts back to $w(L_l) \geq
  |L_l|\lambda_{min}$, exactly the unconditional bound from
  Lemma~\ref{lem:longest-mean-bounds}. Corollary~\ref{cor:max-length-length-bound}
  is unaffected, since it depends only on cycle lengths, which the
  shift preserves. This confirms that when $\lambda_{min} \leq 0$,
  the weight-shifting technique cannot improve upon the weight bounds
  from Lemma~\ref{lem:longest-mean-bounds}.
\end{remark}

\section{Shortest Cycle Bounds}
\label{sec:shortest-results}

Finding the minimum weight directed cycle in a graph without negative
cycles is solvable in $O(nm)$ time \citep{OrSe2017}; with negative
cycles allowed, the problem is NP-hard, since negating all weights
turns it into the max-weight cycle problem. Notably, this
state-of-the-art algorithm explicitly computes $\lambda_{min}$ as a
preprocessing step to reweight the graph
and restrict the search space, empirically confirming that solving the
minimum mean cycle problem is advantageous prior to exact cycle
optimization. Motivated by this connection, we provide the algebraic
dual bounds for shortest cycles below.

We define \textit{min-weight} ($S_w$) and \textit{min-length} ($S_l$)
cycles analogously to longest cycles. The following results parallel
Section~\ref{sec:extensions}.

\begin{lem}[Witness bounds for shortest cycles]
\label{lem:witness-shortest}
For any nonempty set $\mathcal{W}$ of cycles of $G$, $w(S_w) \leq \min\limits_{C \in
  \mathcal{W}} w(C)$ and $|S_l| \leq \min\limits_{C \in \mathcal{W}} |C|$.
\end{lem}

\begin{proof}
$S_w$ minimizes weight and $S_l$ minimizes length over all cycles of
$G$, including those in $\mathcal{W}$.
\end{proof}

\begin{lem}[Cycle mean bounds for shortest cycles]
\label{lem:shortest-mean-bounds}
For any min-weight or min-length cycle $S$,
\begin{equation}
  |S| \lambda_{min} \leq w(S) \leq |S| \lambda_{max}.
\end{equation}
\end{lem}

\begin{proof}
This is a direct application of Lemma~\ref{lem:basic-bounds} to the
specific cycle $S$.
\end{proof}

\subsection{Min-weight Case}
\label{sec:min-weight-case}

Let $S_w$ be a min-weight cycle and $\mathcal{C}_{min}$ be the set of
min-critical cycles $C_{min}$ found in $G$. By definition, $w(S_w) \leq
w(C)$ for all cycles $C$ in $G$.

\begin{lem}[Min-weight key inequality]
  \label{lem:min-weight-key}
  $w(C_{min})(|C_{min}| - |S_w|) \geq 0$.
\end{lem}

\begin{proof}
Since $S_w$ minimizes weight, $w(S_w) \leq w(C_{min}) =
|C_{min}|\lambda_{min}$ for any $C_{min} \in \mathcal{C}_{min}$. From
Lemma~\ref{lem:shortest-mean-bounds}, we have $w(S_w) \geq
|S_w|\lambda_{min}$. Combining these yields $|S_w|\lambda_{min} \leq
|C_{min}|\lambda_{min}$, that is, $\lambda_{min}(|C_{min}| - |S_w|)
\geq 0$. Multiplying by $|C_{min}| > 0$ gives the claim, since
$|C_{min}|\lambda_{min} = w(C_{min})$.
\end{proof}

\begin{thm}[Min-weight cycle bounds]
  \label{thm:min-weight}
  We derive the following three cases from
  Lemma~\ref{lem:min-weight-key}:
  \begin{enumerate}
  \item If $\lambda_{min} > 0$, then
    $|S_w| \leq \min\limits_{C \in \mathcal{C}_{min}} |C|$
    and
    \begin{equation*}
      w(S_w) \leq \min\limits_{C \in \mathcal{C}_{min}} w(C)
      = (\min\limits_{C \in \mathcal{C}_{min}} |C|) \lambda_{min}.
    \end{equation*}
  \item If $\lambda_{min} < 0$, then
    $|S_w| \geq \max\limits_{C \in \mathcal{C}_{min}} |C|$
    and
    \begin{equation*}
      w(S_w) \leq \min\limits_{C \in \mathcal{C}_{min}} w(C)
      = (\max\limits_{C \in \mathcal{C}_{min}} |C|) \lambda_{min}.
    \end{equation*}
  \item If $\lambda_{min} = 0$, then $w(S_w) = 0$.
\end{enumerate}
\end{thm}

\begin{proof}
\mbox{}\par
\begin{itemize}
    \item \textbf{Case 1 ($\lambda_{min} > 0$):} Then $w(C_{min}) =
      |C_{min}|\lambda_{min} > 0$, so dividing
      Lemma~\ref{lem:min-weight-key} by $w(C_{min})$ preserves the
      inequality: $|S_w| \leq |C_{min}|$. Since this holds for
      \textit{any} $C_{min}$, it must hold for the minimum length in
      the set. For the weight, Lemma~\ref{lem:witness-shortest} gives
      $w(S_w) \leq \min\limits_{C \in \mathcal{C}_{min}} w(C)$, and
      since every $C \in \mathcal{C}_{min}$ has $w(C) =
      |C|\lambda_{min}$ with $\lambda_{min} > 0$, the minimum weight is
      attained by the shortest min-critical cycle.
    \item \textbf{Case 2 ($\lambda_{min} < 0$):} Then $w(C_{min}) <
      0$, so dividing Lemma~\ref{lem:min-weight-key} by $w(C_{min})$
      reverses the inequality: $|S_w| \geq |C_{min}|$. Since this
      holds for \textit{any} $C_{min}$, it must hold for the maximum
      length in the set. For the weight,
      Lemma~\ref{lem:witness-shortest} again gives $w(S_w) \leq
      \min\limits_{C \in \mathcal{C}_{min}} w(C)$, and with
      $\lambda_{min} < 0$ the minimum weight is attained by the
      longest min-critical cycle.
    \item \textbf{Case 3 ($\lambda_{min} = 0$):} Then $w(C_{min}) =
      0$, so $w(S_w) \leq 0$. Since $\lambda_{min} = 0$ precludes
      negative cycles, $w(S_w) = 0$.
\end{itemize}
\end{proof}

The weight bounds are the witness bounds of
Lemma~\ref{lem:witness-shortest}; they are tighter than $(\min |C|)
\lambda_{max}$, which also holds but only because $\lambda_{min} \leq
\lambda_{max}$.

\subsection{Min-length Case}
\label{sec:min-length-case}

Let $S_l$ be a min-length cycle and $\mathcal{C}_{max}$ be the set of
max-critical cycles $C_{max}$ found in $G$. By definition, the following
holds.

\begin{lem}[Min-length key inequality]
  \label{lem:min-length-key}
  If $\lambda_{max} \geq 0$, then
  \begin{equation*}
    w(S_l) \leq \min\limits_{C \in \mathcal{C}_{max}} w(C).
  \end{equation*}
\end{lem}

\begin{proof}
For any $C_{max}$, $w(S_l) \leq |S_l|\lambda_{max} \leq
|C_{max}|\lambda_{max} = w(C_{max})$, using $|S_l| \leq |C_{max}|$ and
$\lambda_{max} \geq 0$. The sign condition is necessary: with a
self-loop of weight $-1.5$ and a two-edge cycle whose edges weigh $-1$
each, $\lambda_{max} = -1$, $S_l$ is the self-loop, and $w(S_l) = -1.5
> -2 = w(C_{max})$.
\end{proof}

\begin{thm}[Bounds for min-length cycles]
  \label{thm:minlength-bounds}
\mbox{}\par
\begin{enumerate}
\item If $\lambda_{min} > 0$, then
  $|S_l| \leq \min\limits_{C \in \mathcal{C}_{max}} \lfloor w(C) / \lambda_{min} \rfloor$.
\item If $\lambda_{min} < 0$, the critical cycles yield no
  non-trivial lower bound on $|S_l|$.
\item If $\lambda_{min} = 0$, then
  $w(S_l) \geq 0$.
\end{enumerate}
\end{thm}

\begin{proof}
\mbox{}\par
\begin{itemize}
\item \textbf{Case 1 ($\lambda_{min} > 0$):} Then $\lambda_{max} > 0$,
  so Lemma~\ref{lem:min-length-key} applies and gives
  $|S_l|\lambda_{min} \leq w(S_l) \leq w(C_{max})$ for every $C_{max}$;
  divide by $\lambda_{min}$ and use the integrality of $|S_l|$.
\item \textbf{Case 2 ($\lambda_{min} < 0$):} If $\lambda_{max} \geq
  0$, combining Lemma~\ref{lem:min-length-key} with $w(S_l) \geq
  |S_l|\lambda_{min}$ gives $|S_l| \geq w(C_{max})/\lambda_{min}$, but
  the right-hand side is non-positive since $w(C_{max}) \geq 0$. If
  $\lambda_{max} < 0$, Lemma~\ref{lem:min-length-key} itself may fail,
  as its proof shows; in that example the would-be bound $|S_l| \geq
  w(C_{max})/\lambda_{min} = 4/3$ is violated by $|S_l| = 1$.
\item \textbf{Case 3 ($\lambda_{min} = 0$):} Every cycle has
  non-negative weight.
\end{itemize}
\end{proof}

\section{Heuristic Estimates and Gap Analysis}
\label{sec:heuristic}

While optimum cycle means provide strict bounds, they may be
loose. For heuristic search and when admissibility is not required, an
approximation closer to the true value $\lambda(L)$ is desirable. We
introduce two estimators:
\begin{enumerate}
\item Arithmetic mean $\lambda_{avg} = (\lambda_{min}+\lambda_{max})/2$.
\item Geometric mean $\lambda_{geo} = \sqrt{\lambda_{min} \cdot
  \lambda_{max}}$, defined when $\lambda_{min} > 0$.
\end{enumerate}
Both estimators depend only on the optimum cycle means, which are
uniquely defined and efficiently computable.

We consider the arithmetic and geometric means as inexpensive
estimators of the longest cycle's mean. Since $\lambda_{geo} \leq
\lambda_{avg}$ (Proposition~\ref{prop:am-gm}), the geometric mean is
the lower of the two estimates. Their relative accuracy depends on
both the graph structure and the edge weights; in particular, the
shape of the edge-weight distribution alone does not determine where
in $[\lambda_{min}, \lambda_{max}]$ the mean of a longest cycle
falls. Our experiments compare their performance under uniform and
log-normal weight distributions.

\begin{prop}[Arithmetic-geometric mean inequality]
\label{prop:am-gm}
When $\lambda_{min} > 0$, $\lambda_{geo} \leq \lambda_{avg}$, with
equality if and only if $\lambda_{min} = \lambda_{max}$.
\end{prop}

\begin{proof}
This follows directly from the arithmetic-geometric mean inequality
applied to $\lambda_{min}$ and $\lambda_{max}$.
\end{proof}

\begin{prop}[Absolute approximation error for $\lambda_{avg}$]
  \label{prop:abs-approx-error}
The absolute error of the estimator $\lambda_{avg}$ is bounded as
$|\lambda(L) - \lambda_{avg}| \leq (\lambda_{max} - \lambda_{min})/2$.
\end{prop}

\begin{proof}
By Lemma~\ref{lem:longest-mean-bounds}, $\lambda_{min} \leq \lambda(L)
\leq \lambda_{max}$. Since $\lambda_{avg}$ is the midpoint, the
maximum deviation is half the interval width.
\end{proof}

We define the \textit{weight dispersion}, denoted by $\delta$, as the ratio of
the weight range to the sum of the extreme weights:
\begin{equation}
\delta = \frac{\lambda_{max} - \lambda_{min}}{|\lambda_{max} + \lambda_{min}|},
\end{equation}
assuming $\lambda_{max} + \lambda_{min} \neq 0$. When $\lambda_{max} =
-\lambda_{min}$, the denominator vanishes and $\delta$ is undefined.

\begin{prop}[Relative approximation error for $\lambda_{avg}$]
\label{prop:rel-approx-error}
When $\delta$ is defined, the relative error of the estimator
$\lambda_{avg}$ is bounded by $\delta$:
\begin{equation}
\frac{|\lambda(L) - \lambda_{avg}|}{|\lambda_{avg}|} \leq \delta.
\end{equation}
\end{prop}

\begin{proof}
From Proposition~\ref{prop:abs-approx-error}, the absolute error is
bounded by $(\lambda_{max} - \lambda_{min})/2$.  Substituting the
definition $\lambda_{avg} = (\lambda_{min} + \lambda_{max})/2$ into
the denominator yields:
$$ \frac{|\lambda(L) - \lambda_{avg}|}{|\lambda_{avg}|} \leq
\frac{(\lambda_{max} - \lambda_{min})/2}{|(\lambda_{min} +
  \lambda_{max})/2|} = \frac{\lambda_{max} -
  \lambda_{min}}{|\lambda_{max} + \lambda_{min}|} = \delta.
$$
\end{proof}

As Propositions~\ref{prop:abs-approx-error}
and~\ref{prop:rel-approx-error} establish, $\lambda_{avg}$ provides
rigorous, instance-dependent error bounds. Similar but looser bounds
hold for $\lambda_{geo}$:

\begin{prop}[Absolute approximation error for $\lambda_{geo}$]
  \label{prop:abs-approx-error-geo}
When $\lambda_{min} > 0$, the absolute error of the estimator
$\lambda_{geo}$ is bounded as $|\lambda(L) - \lambda_{geo}| \leq
\lambda_{max} - \lambda_{geo} = \lambda_{max} - \sqrt{\lambda_{min}
  \cdot \lambda_{max}}$.
\end{prop}

\begin{proof}
By Lemma~\ref{lem:longest-mean-bounds}, $\lambda(L) \in
[\lambda_{min}, \lambda_{max}]$, so $|\lambda(L) - \lambda_{geo}|$ is
at most the distance from $\lambda_{geo}$ to the farther endpoint of
this interval. That endpoint is $\lambda_{max}$, because
$(\lambda_{max} - \lambda_{geo}) - (\lambda_{geo} - \lambda_{min}) =
\lambda_{max} + \lambda_{min} - 2\sqrt{\lambda_{min} \cdot
  \lambda_{max}} = (\sqrt{\lambda_{max}} - \sqrt{\lambda_{min}})^2
\geq 0$. The bound is attained when $\lambda(L) = \lambda_{max}$.
\end{proof}

\begin{prop}[Relative approximation error for $\lambda_{geo}$]
  \label{prop:rel-approx-error-geo}
When $\lambda_{min} > 0$, the relative error of the estimator
$\lambda_{geo}$ is bounded as
\begin{equation}
\frac{|\lambda(L) - \lambda_{geo}|}{|\lambda_{geo}|} \leq
\frac{\lambda_{max} - \lambda_{geo}}{\lambda_{geo}} =
\sqrt{\frac{\lambda_{max}}{\lambda_{min}}} - 1.
\end{equation}
\end{prop}

\begin{proof}
Follows directly from Proposition~\ref{prop:abs-approx-error-geo} by
dividing both sides by $|\lambda_{geo}| = \sqrt{\lambda_{min} \cdot
  \lambda_{max}}$.
\end{proof}

Propositions~\ref{prop:rel-approx-error}
and~\ref{prop:rel-approx-error-geo} normalize by the estimate. The
experimental metric $\epsilon$ in
Section~\ref{sec:benchmark-results} normalizes by the true mean
$\lambda(L)$ instead, which the following proposition covers.

\begin{prop}[Relative approximation error with respect to the true mean]
\label{prop:rel-approx-error-true}
When $\lambda_{min} > 0$,
\begin{equation}
\frac{|\lambda(L) - \lambda_{avg}|}{\lambda(L)} \leq
\frac{\lambda_{max} - \lambda_{min}}{2\lambda_{min}}
\quad \text{and} \quad
\frac{|\lambda(L) - \lambda_{geo}|}{\lambda(L)} \leq
\sqrt{\frac{\lambda_{max}}{\lambda_{min}}} - 1.
\end{equation}
\end{prop}

\begin{proof}
For $\lambda_{avg}$, the numerator is at most $(\lambda_{max} -
\lambda_{min})/2$ by Proposition~\ref{prop:abs-approx-error}, and the
denominator is at least $\lambda_{min} > 0$. For $\lambda_{geo}$, if
$\lambda(L) \geq \lambda_{geo}$ the ratio equals $1 -
\lambda_{geo}/\lambda(L) \leq 1 - \lambda_{geo}/\lambda_{max} = 1 -
\sqrt{\lambda_{min}/\lambda_{max}}$; otherwise it equals
$\lambda_{geo}/\lambda(L) - 1 \leq \lambda_{geo}/\lambda_{min} - 1 =
\sqrt{\lambda_{max}/\lambda_{min}} - 1$. Writing $r =
\sqrt{\lambda_{max}/\lambda_{min}} \geq 1$, the second value $r - 1$
dominates the first $1 - 1/r$ because $r + 1/r \geq 2$.
\end{proof}

Thus, while $\lambda_{avg}$ achieves tighter worst-case absolute-error
bounds (by the AM-GM inequality, $\lambda_{max} - \lambda_{geo} \geq
(\lambda_{max} - \lambda_{min})/2$), whereas $\lambda_{geo}$ has the
tighter bound on the relative error with respect to the true mean,
both estimators guarantee bounded approximation error. Since a 
constant-factor approximation is impossible for the longest cycle problem on 
general graphs, these instance-dependent bounds allow practitioners to assess 
solution quality in strongly polynomial time.

Recall that the mean bounds of Lemma~\ref{lem:longest-mean-bounds}
are tight when $\lambda_{min} = \lambda_{max}$, and that $\lambda_{max}
= 0$ determines $w(L_w)$ exactly. Neither condition makes the
witness-based weight and length bounds exact: in a graph consisting of
a directed triangle and a self-loop with all edge weights equal to
$1$, $\lambda_{min} = \lambda_{max} = 1$, yet a self-loop witness gives
only $w(L_w) \geq 1$ and $|L_l| \geq 1$. To quantify the width of the
mean interval in other cases, we define the \textit{bound gap ratio}
$\rho$:
\begin{equation}
  \rho = \frac{\lambda_{max} - \lambda_{min}}{|\lambda_{max}|}
\end{equation}
assuming $\lambda_{max} \neq 0$. A smaller $\rho$ indicates a
narrower mean interval and hence tighter mean bounds and heuristic
estimates; it does not measure the tightness of the witness
bounds. Note that when $\lambda_{min} > 0$, we have $\rho \leq 2\delta$
since $\lambda_{avg} \leq \lambda_{max}$.

\begin{figure}[ht]
  \centering
  \subfloat[Graph]{{\includegraphics[scale=0.15]{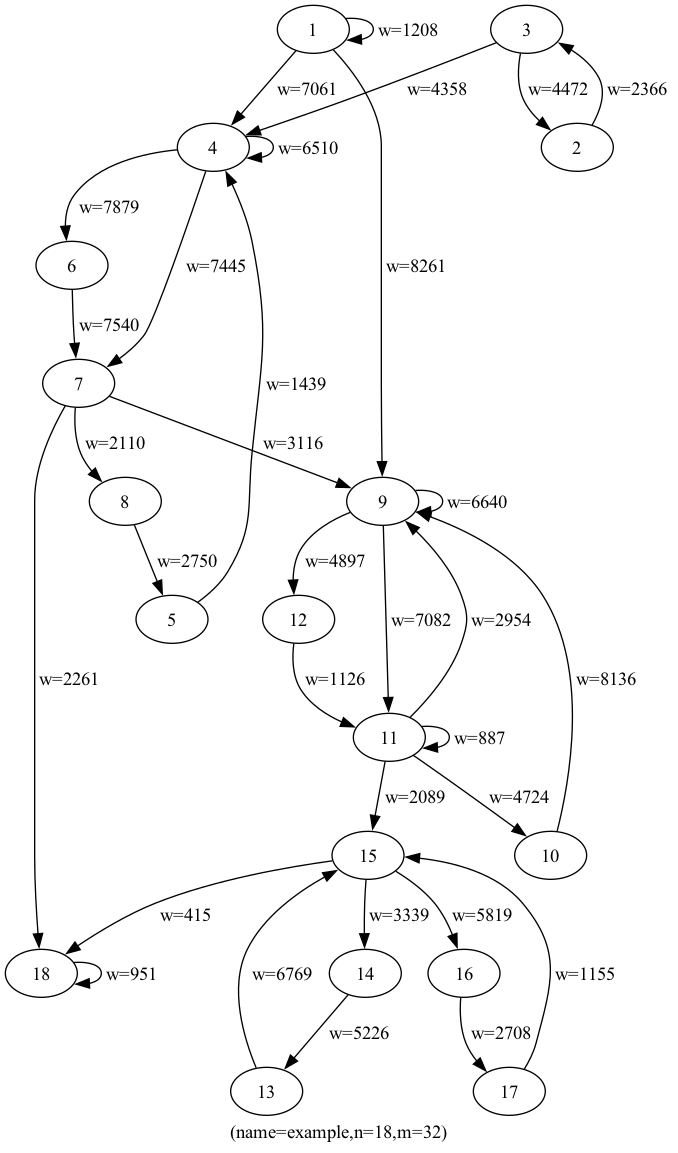}}}
  \qquad
  \subfloat[SCCs]{{\includegraphics[scale=0.12]{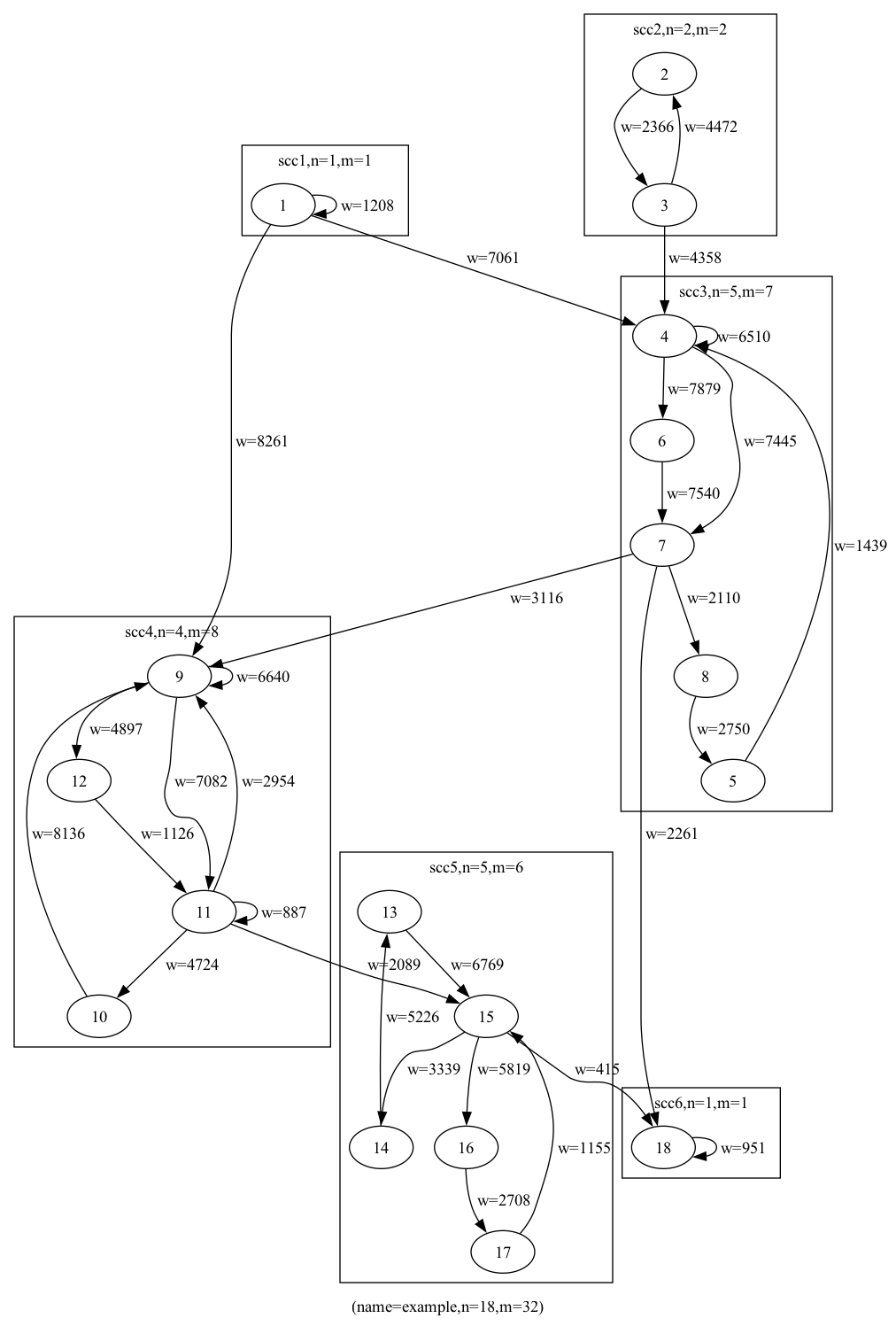}}}
  \caption{(a) Example graph with 18 nodes and 32 edges, and (b) its
    strongly connected components (SCCs).}
  \label{fig:example}
\end{figure}

\section{Illustration}
\label{sec:experiments}

We illustrate our bounds on the example graph in
Figure~\ref{fig:example}. Table~\ref{tab:example1} summarizes the
properties of the SCCs and Table~\ref{tab:example2} lists all cycles
in the example graph.

\begin{table}[ht]
  \centering
  \caption{Properties of each SCC in
    Figure~\ref{fig:example}. Relative errors $\epsilon_{avg}$ and
    $\epsilon_{geo}$ are calculated as
    percentages with respect to $\lambda_{max}$. Bold values indicate
    parameters matching the max-weight/max-length cycles.}
  \label{tab:example1}
  \resizebox{\textwidth}{!}{
      \begin{tabular}{lrrrrrrrrrrr}
        \toprule
        SCC & $n$ & $m$ & $w_{min}$ & $w_{max}$ & $w_{avg}$ & $\lambda_{min}$ & $\lambda_{max}$ & $\lambda_{avg}$ & $\lambda_{geo}$ & $\epsilon_{avg}$ & $\epsilon_{geo}$ \\
        \midrule
        scc1 & $1$ & $1$ & $\mathbf{1,208}$ & $\mathbf{1,208}$ & $1,208.00$ & $1,208.00$ & $1,208.00$ & $1,208.00$ & $1,208.00$ & $0.0\%$ & $0.0\%$ \\
        scc2 & $2$ & $2$ & $2,366$ & $4,472$ & $3,419.00$ & $3,419.00$ & $3,419.00$ & $3,419.00$ & $3,419.00$ & $0.0\%$ & $0.0\%$ \\
        scc3 & $5$ & $7$ & $1,439$ & $7,879$ & $5,096.14$ & $3,436.00$ & $6,510.00$ & $4,973.00$ & $4,729.52$ & $23.6\%$ & $27.3\%$ \\
        scc4 & $4$ & $8$ &   $887$ & $8,136$ & $4,555.75$ &   $887.00$ & $\mathbf{6,647.33}$ & $3,767.17$ & $2,428.21$ & $43.3\%$ & $63.5\%$ \\
        scc5 & $5$ & $6$ & $1,155$ & $6,769$ & $3,646.50$ & $3,227.33$ & $\mathbf{5,111.33}$ & $4,169.33$ & $4,061.52$ & $18.4\%$ & $20.5\%$ \\
        scc6 & $1$ & $1$ & $\mathbf{951}$ & $\mathbf{951}$ & $951.00$ & $951.00$ & $951.00$ & $951.00$ & $951.00$ & $0.0\%$ & $0.0\%$ \\
        \bottomrule
      \end{tabular}
  }
\end{table}

\begin{table}[ht]
  \centering
  \caption{Full cycle enumeration for each strongly connected
    component (SCC). Bold values denote maximum weight or length
    within the component.}
  \label{tab:example2}
  \footnotesize
  \begin{tabular}{llrrcr}
    \toprule
    \multicolumn{1}{c}{SCC} & \multicolumn{1}{c}{Id} & \multicolumn{1}{c}{Edges of cycle C} & \multicolumn{1}{c}{$w(C)$} & \multicolumn{1}{c}{$|C|$} & \multicolumn{1}{c}{$\lambda(C)$} \\
    \midrule
    scc1 & $C_{1}$ & $1 \rightarrow 1$ & $\mathbf{1,208}$ & $\mathbf{1}$ & $1,208.00$ \\
    \midrule 
    scc2 & $C_{2}$ & $2 \rightarrow 3 \rightarrow 2$ & $\mathbf{6,838}$ & $\mathbf{2}$ & $3,419.00$ \\
    \midrule 
         & $C_{3}$ & $4 \rightarrow 4$ & $6,510$ & $1$ & $6,510.00$ \\
    scc3 & $C_{4}$ & $4 \rightarrow 6 \rightarrow 7 \rightarrow 8 \rightarrow 5 \rightarrow 4$ & $\mathbf{21,718}$ & $\mathbf{5}$ & $4,343.60$ \\
         & $C_{5}$ & $4 \rightarrow 7 \rightarrow 8 \rightarrow 5 \rightarrow 4$ & $13,744$ & $4$ & $3,436.00$ \\
    \midrule 
         & $C_{6}$ & $9 \rightarrow 11 \rightarrow 10 \rightarrow 9$ & $\mathbf{19,942}$ & $3$ & $6,647.33$ \\
         & $C_{7}$ & $9 \rightarrow 9$ & $6,640$ & $1$ & $6,640.00$ \\
    scc4 & $C_{8}$ & $9 \rightarrow 11 \rightarrow 9$ & $10,036$ & $2$ & $5,018.00$ \\
         & $C_{9}$ & $9 \rightarrow 12 \rightarrow 11 \rightarrow 10 \rightarrow 9$ & $18,883$ & $\mathbf{4}$ & $4,720.75$ \\
         & $C_{10}$ & $9 \rightarrow 12 \rightarrow 11 \rightarrow 9$ & $8,977$ & $3$ & $2,992.33$ \\
         & $C_{11}$ & $11 \rightarrow 11$ & $887$ & $1$ & $887.00$ \\
    \midrule 
    scc5 & $C_{12}$ & $13 \rightarrow 15 \rightarrow 14 \rightarrow 13$ & $\mathbf{15,334}$ & $\mathbf{3}$ & $5,111.33$ \\
         & $C_{13}$ & $15 \rightarrow 16 \rightarrow 17 \rightarrow 15$ & $9,682$ & $\mathbf{3}$ & $3,227.33$ \\
    \midrule 
    scc6 & $C_{14}$ & $18 \rightarrow 18$ & $\mathbf{951}$ & $\mathbf{1}$ & $951.00$ \\
    \bottomrule
  \end{tabular}
\end{table}

In trivial SCCs (scc1, scc2, scc6) with a single cycle, the cycle is
both longest and shortest. In scc3, the max-weight and max-length
cycle is $C_4$, while $C_3$ is max-critical and $C_5$ is min-critical.

Applying our bounds to scc3 (where $\lambda_{max} = 6510 > 0$):
\begin{itemize}
\item Mean bound: $3436 \leq \lambda(L_w) \leq 6510$. (True: 4343.60.)
\item Weight bound (Theorem~\ref{thm:max-weight}):
  \begin{equation*}
    w(L_w) \geq B_w = \max\{w(C_5), w(C_3)\} = \max\{13744, 6510\} = 13744.
  \end{equation*}
  (True: 21718.) The weaker bound $|C_{max}|\lambda_{min} = 1 \times
  3436 = 3436$ discards the min-critical witness.
\item Length bound (Theorem~\ref{thm:max-weight}): $|L_w| \geq \lceil
  B_w/\lambda_{max} \rceil = \lceil 13744/6510 \rceil = 3$. (True: 5.)
\item Max-length bounds (Theorem~\ref{thm:max-length} and
  Corollary~\ref{cor:max-length-length-bound}):
  \begin{align*}
    |L_l| &\geq B_l = \max\{|C_5|, |C_3|\} = 4, \\
    w(L_l) &\geq B_l\lambda_{min} = 4 \times 3436 = 13744.
  \end{align*}
  (True: 5 and 21718.)
\end{itemize}

In scc4, the max-weight cycle $C_6$ coincides with $C_{max}$, so the
bound is tight: $w(L_w) = w(C_{max}) = 19942$.

\begin{table}[ht]
\centering
\caption{Basic properties of the ISCAS benchmark graphs used in the
  evaluation.}
\label{tab:iscas_properties}
\footnotesize
\begin{tabular}{lrrr}
\toprule
\multicolumn{1}{c}{Graph} & \multicolumn{1}{c}{$n$} & \multicolumn{1}{c}{$m$} & \multicolumn{1}{c}{\#Cycles} \\
\midrule
mm30a & 2,059 & 3,912 & 540 \\
ecc & 1,618 & 2,843 & 493 \\
sbc & 1,147 & 1,791 & 19,573 \\
mm9b & 777 & 1,452 & 84,851 \\
mm9a & 631 & 1,182 & 171 \\
s838 & 665 & 941 & 56 \\
s641 & 477 & 612 & 853 \\
s526 & 318 & 576 & 107 \\
s526n & 292 & 560 & 103 \\
s444 & 315 & 503 & 296 \\
s400 & 287 & 462 & 148 \\
s382 & 273 & 438 & 106 \\
s349 & 278 & 395 & 702 \\
s344 & 274 & 388 & 690 \\
mm4a & 170 & 454 & 136 \\
s420 & 104 & 178 & 4 \\
s208 & 83 & 119 & 10 \\
s27 & 55 & 87 & 7 \\
\midrule
\textbf{Average} & 545.7 & 938.5 & 6,047.0 \\
\textbf{Median} & 303.5 & 531.5 & 159.5 \\
\textbf{StDev} & 549.3 & 1,011.0 & 20,185.9 \\
\bottomrule
\end{tabular}
\end{table}

\section{Experimental Results on Benchmark Graphs}
\label{sec:benchmark-results}

\subsection{Methodology}

To evaluate our bounds on realistic and larger graphs, we utilized the
ISCAS benchmark circuits \citep{BrFu1985,BrBrKo1989,HaYaHa1999}, a
standard test suite in electronic design automation. Due to the high
computational cost of the exact cycle enumeration required to
establish ground truth \citep{HaJa08}, we selected a subset of smaller
instances from the \texttt{iscas-small} directory of the benchmark
repository \citep{Da14}. These graphs, listed in
Table~\ref{tab:iscas_properties}, range from 55 to 2,059 nodes and
contain between 4 and 84,851 cycles.

Since the original ISCAS benchmarks are unweighted logic circuits, we
assigned integer weights to edges using two different distributions to
test the robustness of our bounds:
\begin{enumerate}
    \item \textbf{Uniform Distribution:} Edge weights were generated
      uniformly at random in the range $[1, 3000]$. This represents a
      symmetric distribution with no skew.
    \item \textbf{Log-Normal Distribution:} Edge weights $W$ were
      generated with $\ln W \sim N(\mu, \sigma^2)$, where $\mu =
      \ln(3000)/2$ and $\sigma = \ln(3000)/6$, so that $e^{\mu -
        3\sigma} = 1$ and $e^{\mu + 3\sigma} = 3000$; that is, the
      range $[1, 3000]$ corresponds to the $6\sigma$ interval of the
      underlying normal variable, which covers $\approx 99.7\%$ of the
      samples. Samples were rounded to integer weights. This models
      realistic circuit delay distributions where delays are skewed
      and strictly positive.
\end{enumerate}

Ground truth for the longest cycles ($L_w, L_l$) was established via
explicit enumeration using the algorithm of Hawick and James
\citep{HaJa08}, implemented in C based on \citep{Sc08}. Optimum cycle
means were computed using Howard's algorithm with the publicly
available implementation from \citep{Da15}, which was also used in
\citep{Da04}.

Regarding computational efficiency, the optimum cycle means were
computed in under one second for all benchmarks using Howard's
algorithm. The cycle enumeration (used only for validation) was the
bottleneck, taking several minutes for graphs with tens of thousands
of cycles; it did not finish after hours for larger ISCAS benchmark
graphs. In practice, the bounds would be used precisely to avoid such
enumeration.

\subsection{Evaluation Metrics}

To quantify the performance of our theoretical framework, we define two
metrics. First, the \textit{bound relative error} ($\Delta$) measures
the looseness of our strict lower bounds relative to the true ground
truth values. For a true value $V$ and a derived lower bound $V_{lb}$, we
define $\Delta = (V - V_{lb})/V$. A value of $\Delta = 0\%$ indicates a
tight bound, while values approaching $100\%$ indicate looseness.
Specifically, for max-weight cycles we evaluate:
\begin{equation}
  \Delta_{w(L_w)} = \frac{w(L_w) - B_w}{w(L_w)}, \quad
  \Delta_{|L_w|} = \frac{|L_w| - \lceil B_w/\lambda_{max} \rceil}{|L_w|},
\end{equation}
and for max-length cycles:
\begin{equation}
  \Delta_{w(L_l)} = \frac{w(L_l) - B_l\lambda_{min}}{w(L_l)}, \quad
  \Delta_{|L_l|} = \frac{|L_l| - B_l}{|L_l|},
\end{equation}
where $B_w = \max\{w(C_{min}), w(C_{max})\}$ and $B_l = \max\{|C_{min}|,
|C_{max}|\}$ are computed from the single min-critical and
max-critical witness reported for each graph
(Theorems~\ref{thm:max-weight} and~\ref{thm:max-length} and
Corollary~\ref{cor:max-length-length-bound}).
Second, the \textit{heuristic approximation error} ($\epsilon$)
measures the accuracy of our cycle mean estimators ($\lambda_{avg},
\lambda_{geo}$) compared to the true mean of the longest cycle
$\lambda(L)$. We define this as the standard relative error:
\begin{equation}
  \epsilon = \frac{|\lambda(L) - \lambda_{est}|}{|\lambda(L)|},
\end{equation}
where $\lambda_{est} \in \{\lambda_{avg}, \lambda_{geo}\}$.

\subsection{Results with Uniform Weights}

Tables~\ref{tab:iscas_results} and~\ref{tab:heuristic_results}
summarize the results for the uniform distribution.

\begin{table}[ht]
\centering
\caption{Experimental results on ISCAS benchmark graphs with
  \textbf{uniform} weights. The table compares properties of the
  max-weight ($L_w$) and max-length ($L_l$) cycles against the
  min/max-critical cycles ($C_{min}, C_{max}$). $\Delta$ columns show
  the relative error of the derived lower bounds compared to the true
  values. $\rho$ is the bound gap ratio.}
\label{tab:iscas_results}
\resizebox{\textwidth}{!}{%
\begin{tabular}{lrrrrrrrrrrrrrrrrr}
\toprule
\multicolumn{1}{c}{Graph} & \multicolumn{1}{c}{$w(L_w)$} & \multicolumn{1}{c}{$|L_w|$} & \multicolumn{1}{c}{$\lambda(L_w)$} & \multicolumn{1}{c}{$w(L_l)$} & \multicolumn{1}{c}{$|L_l|$} & \multicolumn{1}{c}{$\lambda(L_l)$} & \multicolumn{1}{c}{$w(C_{min})$} & \multicolumn{1}{c}{$|C_{min}|$} & \multicolumn{1}{c}{$w(C_{max})$} & \multicolumn{1}{c}{$|C_{max}|$} & \multicolumn{1}{c}{$\lambda_{min}$} & \multicolumn{1}{c}{$\lambda_{max}$} & \multicolumn{1}{c}{$\Delta_{w(L_w)}$} & \multicolumn{1}{c}{$\Delta_{|L_w|}$} & \multicolumn{1}{c}{$\Delta_{w(L_l)}$} & \multicolumn{1}{c}{$\Delta_{|L_l|}$} & \multicolumn{1}{c}{$\rho$} \\
\midrule
mm30a & 134,012 & 77 & 1740.4 & 134,012 & 77 & 1740.4 & 7,213 & 10 & 21,057 & 10 & 721.3 & 2105.7 & 84.3\% & 87.0\% & 94.6\% & 87.0\% & 0.657 \\
ecc & 270,922 & 167 & 1622.3 & 264,044 & 167 & 1581.1 & 1,579 & 3 & 7,527 & 3 & 526.3 & 2509.0 & 97.2\% & 98.2\% & 99.4\% & 98.2\% & 0.790 \\
sbc & 110,459 & 71 & 1555.8 & 109,530 & 72 & 1521.2 & 7,825 & 10 & 12,529 & 6 & 782.5 & 2088.2 & 88.7\% & 91.5\% & 92.9\% & 86.1\% & 0.625 \\
mm9b & 130,409 & 81 & 1610.0 & 126,037 & 85 & 1482.8 & 2,899 & 5 & 10,643 & 5 & 579.8 & 2128.6 & 91.8\% & 93.8\% & 97.7\% & 94.1\% & 0.728 \\
mm9a & 50,642 & 31 & 1633.6 & 47,520 & 35 & 1357.7 & 4,273 & 5 & 10,109 & 5 & 854.6 & 2021.8 & 80.0\% & 83.9\% & 91.0\% & 85.7\% & 0.577 \\
s838 & 11,362 & 6 & 1893.7 & 10,279 & 6 & 1713.2 & 2,102 & 5 & 10,278 & 5 & 420.4 & 2055.6 & 9.5\% & 16.7\% & 79.6\% & 16.7\% & 0.795 \\
s641 & 263,663 & 174 & 1515.3 & 263,663 & 174 & 1515.3 & 10,966 & 10 & 23,006 & 12 & 1096.6 & 1917.2 & 91.3\% & 93.1\% & 95.0\% & 93.1\% & 0.428 \\
s526 & 37,844 & 22 & 1720.2 & 35,803 & 23 & 1556.7 & 1,831 & 3 & 10,932 & 5 & 610.3 & 2186.4 & 71.1\% & 77.3\% & 91.5\% & 78.3\% & 0.721 \\
s526n & 40,147 & 22 & 1824.9 & 39,721 & 23 & 1727.0 & 3,922 & 5 & 11,530 & 5 & 784.4 & 2306.0 & 71.3\% & 77.3\% & 90.1\% & 78.3\% & 0.660 \\
s444 & 67,685 & 45 & 1504.1 & 64,885 & 45 & 1441.9 & 7,447 & 10 & 11,391 & 5 & 744.7 & 2278.2 & 83.2\% & 88.9\% & 88.5\% & 77.8\% & 0.673 \\
s400 & 58,286 & 35 & 1665.3 & 54,202 & 36 & 1505.6 & 5,101 & 7 & 12,178 & 6 & 728.7 & 2029.7 & 79.1\% & 82.9\% & 90.6\% & 80.6\% & 0.641 \\
s382 & 53,361 & 33 & 1617.0 & 50,394 & 34 & 1482.2 & 6,450 & 7 & 15,950 & 7 & 921.4 & 2278.6 & 70.1\% & 78.8\% & 87.2\% & 79.4\% & 0.596 \\
s349 & 130,251 & 87 & 1497.1 & 128,554 & 89 & 1444.4 & 5,008 & 6 & 10,699 & 5 & 834.7 & 2139.8 & 91.8\% & 94.3\% & 96.1\% & 93.3\% & 0.610 \\
s344 & 139,510 & 89 & 1567.5 & 139,510 & 89 & 1567.5 & 7,991 & 9 & 16,644 & 8 & 887.9 & 2080.5 & 88.1\% & 91.0\% & 94.3\% & 89.9\% & 0.573 \\
mm4a & 17,524 & 10 & 1752.4 & 12,427 & 10 & 1242.7 & 6,793 & 8 & 15,399 & 8 & 849.1 & 1924.9 & 12.1\% & 20.0\% & 45.3\% & 20.0\% & 0.559 \\
s420 & 7,976 & 6 & 1329.3 & 7,263 & 6 & 1210.5 & 4,848 & 5 & 7,976 & 6 & 969.6 & 1329.3 & 0.0\% & 0.0\% & 19.9\% & 0.0\% & 0.271 \\
s208 & 9,990 & 5 & 1998.0 & 9,202 & 6 & 1533.7 & 7,318 & 6 & 9,990 & 5 & 1219.7 & 1998.0 & 0.0\% & 0.0\% & 20.5\% & 0.0\% & 0.390 \\
s27 & 15,539 & 10 & 1553.9 & 15,539 & 10 & 1553.9 & 14,236 & 10 & 8,443 & 5 & 1423.6 & 1688.6 & 8.4\% & 10.0\% & 8.4\% & 0.0\% & 0.157 \\
\midrule
\textbf{Average} & 86,087.9 & 53.9 & 1644.5 & 84,032.5 & 54.8 & 1509.9 & 5,989.0 & 6.9 & 12,571.2 & 6.2 & 830.9 & 2059.2 & 62.1\% & 65.8\% & 76.8\% & 64.4\% & 0.581 \\
\textbf{Median} & 55,823.5 & 34.0 & 1619.6 & 52,298.0 & 35.5 & 1518.3 & 5,775.5 & 6.5 & 11,161.5 & 5.0 & 809.5 & 2084.3 & 79.6\% & 83.4\% & 90.8\% & 80.0\% & 0.618 \\
\textbf{StDev} & 80,581.5 & 51.7 & 158.8 & 80,222.4 & 51.6 & 143.4 & 3,236.5 & 2.5 & 4,277.3 & 2.1 & 244.0 & 254.9 & 36.7\% & 36.8\% & 30.3\% & 37.2\% & 0.171 \\
\bottomrule
\end{tabular}%
}
\end{table}

\begin{table}[ht]
\centering
\caption{Experimental results on ISCAS benchmark graphs with
  \textbf{uniform} weights. The table compares heuristic
  approximations $\lambda_{avg}$ and $\lambda_{geo}$ against cycle
  means of the max-weight ($L_w$) and max-length ($L_l$)
  cycles. Relative errors $\epsilon$ are calculated as percentages.}
\label{tab:heuristic_results}
\resizebox{\textwidth}{!}{%
\begin{tabular}{@{}lrrrrrrrr@{}}
\toprule
\multicolumn{1}{c}{Graph} & \multicolumn{1}{c}{$\lambda_{avg}$} & \multicolumn{1}{c}{$\lambda_{geo}$} & \multicolumn{1}{c}{$\epsilon_{avg}(L_w)$} & \multicolumn{1}{c}{$\epsilon_{geo}(L_w)$} & \multicolumn{1}{c}{$\epsilon_{avg}(L_l)$} & \multicolumn{1}{c}{$\epsilon_{geo}(L_l)$} \\
\midrule
mm30a & 1413.5 & 1232.4 & 18.8\% & 29.2\% & 18.8\% & 29.2\% \\
ecc & 1517.7 & 1149.2 & 6.4\% & 29.2\% & 4.0\% & 27.3\% \\
sbc & 1435.3 & 1278.3 & 7.7\% & 17.8\% & 5.6\% & 16.0\% \\
mm9b & 1354.2 & 1110.9 & 15.9\% & 31.0\% & 8.7\% & 25.1\% \\
mm9a & 1438.2 & 1314.5 & 12.0\% & 19.5\% & 5.9\% & 3.2\% \\
s838 & 1238.0 & 929.6 & 34.6\% & 50.9\% & 27.7\% & 45.7\% \\
s641 & 1506.9 & 1450.0 & 0.6\% & 4.3\% & 0.6\% & 4.3\% \\
s526 & 1398.4 & 1155.2 & 18.7\% & 32.8\% & 10.2\% & 25.8\% \\
s526n & 1545.2 & 1344.9 & 15.3\% & 26.3\% & 10.5\% & 22.1\% \\
s444 & 1511.4 & 1302.5 & 0.5\% & 13.4\% & 4.8\% & 9.7\% \\
s400 & 1379.2 & 1216.2 & 17.2\% & 27.0\% & 8.4\% & 19.2\% \\
s382 & 1600.0 & 1449.0 & 1.1\% & 10.4\% & 7.9\% & 2.2\% \\
s349 & 1487.2 & 1336.4 & 0.7\% & 10.7\% & 3.0\% & 7.5\% \\
s344 & 1484.2 & 1359.1 & 5.3\% & 13.3\% & 5.3\% & 13.3\% \\
mm4a & 1387.0 & 1278.5 & 20.9\% & 27.0\% & 11.6\% & 2.9\% \\
s420 & 1149.5 & 1135.3 & 13.5\% & 14.6\% & 5.0\% & 6.2\% \\
s208 & 1608.8 & 1561.1 & 19.5\% & 21.9\% & 4.9\% & 1.8\% \\
s27 & 1556.1 & 1550.4 & 0.1\% & 0.2\% & 0.1\% & 0.2\% \\
\midrule
\textbf{Average} & 1445.0 & 1286.3 & 11.6\% & 21.1\% & 8.0\% & 14.5\% \\
\textbf{Median} & 1461.2 & 1290.5 & 12.7\% & 20.7\% & 5.8\% & 11.5\% \\
\textbf{StDev} & 118.7 & 160.1 & 9.5\% & 12.0\% & 6.6\% & 12.6\% \\
\bottomrule
\end{tabular}%
}
\end{table}

As shown in Tables~\ref{tab:iscas_results}
and~\ref{tab:heuristic_results}, the uniform distribution results
validate our theoretical framework. A key observation is that
$\lambda(L_w)$ and $\lambda(L_l)$ always fall within
$[\lambda_{min}, \lambda_{max}]$, confirming
Lemma~\ref{lem:longest-mean-bounds}. The average bound gap ratio
$\rho$ is 0.581, indicating that on average, the interval spans
roughly 58\% of the maximum possible mean value. All lower bounds
except that on $w(L_l)$ are tight ($\Delta = 0\%$) on the small
circuits s420 and s208, where the max-critical witness is itself the
max-weight cycle and a critical witness attains the maximum length,
but the median errors of 80--91\% on the full set reflect how much
longer the longest cycles are than the critical cycles.

The use of uniformly distributed random weights results in a symmetric
distribution of edge weights, which suggests an empirical
interpretation of these results. First, the strong correlation between
max-weight and max-length cycles is expected: with independent
identically distributed weights, total cycle weight grows roughly
linearly with length, making the longest cycle ($L_l$) a natural
candidate for the heaviest cycle ($L_w$). Second, the accuracy of the
heuristic $\lambda_{avg}$ is consistent with long cycles having means
close to the expected edge weight, which for a symmetric distribution
lies at the midpoint of the weight range. This reasoning is heuristic
rather than a guarantee: $L_w$ is selected because its total weight is
maximal, which biases its mean upward, and the symmetry of the edge
distribution does not by itself place the midpoint of the extreme
cycle means at the population mean. Empirically, however, the median
heuristic errors of $\lambda_{avg}$ are low, 12.7\% ($L_w$) and 5.8\%
($L_l$), whereas $\lambda_{geo}$ underestimates the cycle means for
uniform weights, yielding higher median errors of 20.7\% ($L_w$) and
11.5\% ($L_l$).

\subsection{Results with Log-Normal Weights}

Tables~\ref{tab:iscas_results_log} and~\ref{tab:heuristic_results_log}
present the results for the log-normal distribution.

\begin{table}[ht]
\centering
\caption{Experimental results on ISCAS benchmark graphs with
  \textbf{log-normal} weights. The table compares properties of the
  max-weight ($L_w$) and max-length ($L_l$) cycles against the
  min/max-critical cycles ($C_{min}, C_{max}$). $\Delta$ columns show
  the relative error of the derived lower bounds compared to the true
  values. $\rho$ is the bound gap ratio.}
\label{tab:iscas_results_log}
\resizebox{\textwidth}{!}{%
\begin{tabular}{lrrrrrrrrrrrrrrrrr}
\toprule
\multicolumn{1}{c}{Graph} & \multicolumn{1}{c}{$w(L_w)$} & \multicolumn{1}{c}{$|L_w|$} & \multicolumn{1}{c}{$\lambda(L_w)$} & \multicolumn{1}{c}{$w(L_l)$} & \multicolumn{1}{c}{$|L_l|$} & \multicolumn{1}{c}{$\lambda(L_l)$} & \multicolumn{1}{c}{$w(C_{min})$} & \multicolumn{1}{c}{$|C_{min}|$} & \multicolumn{1}{c}{$w(C_{max})$} & \multicolumn{1}{c}{$|C_{max}|$} & \multicolumn{1}{c}{$\lambda_{min}$} & \multicolumn{1}{c}{$\lambda_{max}$} & \multicolumn{1}{c}{$\Delta_{w(L_w)}$} & \multicolumn{1}{c}{$\Delta_{|L_w|}$} & \multicolumn{1}{c}{$\Delta_{w(L_l)}$} & \multicolumn{1}{c}{$\Delta_{|L_l|}$} & \multicolumn{1}{c}{$\rho$} \\
\midrule
mm30a & 13,515 & 62 & 218.0 & 7,180 & 77 & 93.2 & 336 & 10 & 8,554 & 26 & 33.6 & 329.0 & 36.7\% & 58.1\% & 87.8\% & 66.2\% & 0.898 \\
ecc & 24,045 & 166 & 144.8 & 24,001 & 167 & 143.7 & 39 & 3 & 3,295 & 5 & 13.0 & 659.0 & 86.3\% & 97.0\% & 99.7\% & 97.0\% & 0.980 \\
sbc & 12,572 & 52 & 241.8 & 8,476 & 72 & 117.7 & 323 & 9 & 7,121 & 10 & 35.9 & 712.1 & 43.4\% & 80.8\% & 95.8\% & 86.1\% & 0.950 \\
mm9b & 13,818 & 82 & 168.5 & 11,449 & 85 & 134.7 & 189 & 5 & 2,944 & 13 & 37.8 & 226.5 & 78.7\% & 84.1\% & 95.7\% & 84.7\% & 0.833 \\
mm9a & 5,690 & 24 & 237.1 & 4,014 & 35 & 114.7 & 175 & 5 & 2,052 & 5 & 35.0 & 410.4 & 63.9\% & 79.2\% & 95.6\% & 85.7\% & 0.915 \\
s838 & 3,573 & 5 & 714.6 & 476 & 6 & 79.3 & 71 & 4 & 3,573 & 5 & 17.8 & 714.6 & 0.0\% & 0.0\% & 81.4\% & 16.7\% & 0.975 \\
s641 & 28,826 & 162 & 177.9 & 23,747 & 174 & 136.5 & 736 & 12 & 4,364 & 10 & 61.3 & 436.4 & 84.9\% & 93.8\% & 96.9\% & 93.1\% & 0.859 \\
s526 & 4,671 & 18 & 259.5 & 3,098 & 23 & 134.7 & 119 & 5 & 2,024 & 5 & 23.8 & 404.8 & 56.7\% & 72.2\% & 96.2\% & 78.3\% & 0.941 \\
s526n & 3,834 & 21 & 182.6 & 2,551 & 23 & 110.9 & 100 & 5 & 1,590 & 3 & 20.0 & 530.0 & 58.5\% & 85.7\% & 96.1\% & 78.3\% & 0.962 \\
s444 & 7,907 & 45 & 175.7 & 4,455 & 45 & 99.0 & 260 & 6 & 3,085 & 8 & 43.3 & 385.6 & 61.0\% & 82.2\% & 92.2\% & 82.2\% & 0.888 \\
s400 & 7,264 & 35 & 207.5 & 3,993 & 36 & 110.9 & 266 & 6 & 3,088 & 6 & 44.3 & 514.7 & 57.5\% & 82.9\% & 93.3\% & 83.3\% & 0.914 \\
s382 & 5,978 & 33 & 181.2 & 4,874 & 34 & 143.4 & 288 & 10 & 3,511 & 7 & 28.8 & 501.6 & 41.3\% & 78.8\% & 94.1\% & 70.6\% & 0.943 \\
s349 & 8,995 & 89 & 101.1 & 8,995 & 89 & 101.1 & 259 & 5 & 1,782 & 6 & 51.8 & 297.0 & 80.2\% & 93.3\% & 96.5\% & 93.3\% & 0.826 \\
s344 & 14,128 & 83 & 170.2 & 11,917 & 89 & 133.9 & 417 & 8 & 3,188 & 7 & 52.1 & 455.4 & 77.4\% & 91.6\% & 96.5\% & 91.0\% & 0.886 \\
mm4a & 5,741 & 9 & 637.9 & 1,458 & 10 & 145.8 & 215 & 9 & 5,741 & 9 & 23.9 & 637.9 & 0.0\% & 0.0\% & 85.3\% & 10.0\% & 0.963 \\
s420 & 460 & 6 & 76.7 & 382 & 6 & 63.7 & 242 & 5 & 460 & 6 & 48.4 & 76.7 & 0.0\% & 0.0\% & 24.0\% & 0.0\% & 0.369 \\
s208 & 3,766 & 6 & 627.7 & 3,766 & 6 & 627.7 & 522 & 5 & 3,766 & 6 & 104.4 & 627.7 & 0.0\% & 0.0\% & 83.4\% & 0.0\% & 0.834 \\
s27 & 845 & 10 & 84.5 & 586 & 10 & 58.6 & 265 & 5 & 538 & 5 & 53.0 & 107.6 & 36.3\% & 50.0\% & 54.8\% & 50.0\% & 0.507 \\
\midrule
\textbf{Average} & 9,201.6 & 50.4 & 256.0 & 6,967.7 & 54.8 & 141.6 & 267.9 & 6.5 & 3,370.9 & 7.9 & 40.5 & 445.9 & 47.9\% & 62.8\% & 87.0\% & 64.8\% & 0.858 \\
\textbf{Median} & 6,621.0 & 34.0 & 181.9 & 4,234.5 & 35.5 & 116.2 & 259.5 & 5.0 & 3,138.0 & 6.0 & 36.8 & 445.9 & 57.1\% & 80.0\% & 94.9\% & 80.2\% & 0.906 \\
\textbf{StDev} & 7,584.7 & 49.7 & 193.3 & 7,080.6 & 51.6 & 124.2 & 167.0 & 2.5 & 2,084.7 & 5.1 & 21.0 & 190.2 & 30.6\% & 36.4\% & 18.8\% & 34.0\% & 0.162 \\
\bottomrule
\end{tabular}%
}
\end{table}

\begin{table}[ht]
\centering
\caption{Experimental results on ISCAS benchmark graphs with
  \textbf{log-normal} weights. The table compares heuristic
  approximations $\lambda_{avg}$ and $\lambda_{geo}$ against cycle
  means of the max-weight ($L_w$) and max-length ($L_l$)
  cycles. Relative errors $\epsilon$ are calculated as percentages.}
\label{tab:heuristic_results_log}
\resizebox{\textwidth}{!}{%
\begin{tabular}{@{}lrrrrrr@{}}
\toprule
\multicolumn{1}{c}{Graph} & \multicolumn{1}{c}{$\lambda_{avg}$} & \multicolumn{1}{c}{$\lambda_{geo}$} & \multicolumn{1}{c}{$\epsilon_{avg}(L_w)$} & \multicolumn{1}{c}{$\epsilon_{geo}(L_w)$} & \multicolumn{1}{c}{$\epsilon_{avg}(L_l)$} & \multicolumn{1}{c}{$\epsilon_{geo}(L_l)$} \\
\midrule
mm30a & 181.3 & 105.1 & 16.8\% & 51.8\% & 94.4\% & 12.8\% \\
ecc & 336.0 & 92.6 & 132.0\% & 36.1\% & 133.8\% & 35.6\% \\
sbc & 374.0 & 159.9 & 54.7\% & 33.9\% & 217.7\% & 35.8\% \\
mm9b & 132.1 & 92.5 & 21.6\% & 45.1\% & 1.9\% & 31.3\% \\
mm9a & 222.7 & 119.8 & 6.1\% & 49.4\% & 94.2\% & 4.5\% \\
s838 & 366.2 & 112.6 & 48.8\% & 84.2\% & 361.6\% & 42.0\% \\
s641 & 248.9 & 163.6 & 39.9\% & 8.1\% & 82.4\% & 19.9\% \\
s526 & 214.3 & 98.2 & 17.4\% & 62.2\% & 59.1\% & 27.1\% \\
s526n & 275.0 & 103.0 & 50.6\% & 43.6\% & 147.9\% & 7.2\% \\
s444 & 214.5 & 129.3 & 22.1\% & 26.4\% & 116.6\% & 30.6\% \\
s400 & 279.5 & 151.1 & 34.7\% & 27.2\% & 152.0\% & 36.2\% \\
s382 & 265.2 & 120.2 & 46.4\% & 33.7\% & 85.0\% & 16.2\% \\
s349 & 174.4 & 124.0 & 72.6\% & 22.7\% & 72.6\% & 22.7\% \\
s344 & 253.8 & 154.1 & 49.1\% & 9.5\% & 89.5\% & 15.1\% \\
mm4a & 330.9 & 123.4 & 48.1\% & 80.6\% & 126.9\% & 15.3\% \\
s420 & 62.5 & 60.9 & 18.4\% & 20.5\% & 1.8\% & 4.3\% \\
s208 & 366.0 & 256.0 & 41.7\% & 59.2\% & 41.7\% & 59.2\% \\
s27 & 80.3 & 75.5 & 5.0\% & 10.6\% & 37.0\% & 28.9\% \\
\midrule
\textbf{Average} & 243.2 & 124.5 & 40.3\% & 39.2\% & 106.5\% & 24.7\% \\
\textbf{Median} & 251.3 & 120.0 & 40.8\% & 35.0\% & 91.9\% & 24.9\% \\
\textbf{StDev} & 93.5 & 43.3 & 29.4\% & 22.6\% & 83.5\% & 14.4\% \\
\bottomrule
\end{tabular}%
}
\end{table}

Under log-normal weights, the witness bound $w(L_w) \geq B_w$ is
noticeably tighter (median $\Delta_{w(L_w)}$ of 57.1\%) than under
uniform weights, because the heavy tail lets a max-critical witness
carry a large share of the max-weight cycle's total, and the two
coincide on s838, mm4a, s420, and s208.

The log-normal results highlight the sensitivity of the heuristics
to the weight distribution. The log-normal distribution is
right-skewed, whereas the uniform distribution is symmetric. As seen in
Table~\ref{tab:heuristic_results_log}, the arithmetic estimator
$\lambda_{avg}$ has median relative errors of 40.8\% for max-weight
cycles and a drastic 91.9\% for max-length cycles. The direction of
its error is not uniform: it mostly overestimates the cycle means, but
on s838, for example, $\lambda_{avg} = 366.2$ while $\lambda(L_w) =
714.6$, because a max-weight cycle is selected by its weights and its
mean need not track the population mean.

The geometric estimator $\lambda_{geo}$ reduces these medians to
35.0\% for max-weight cycles, a modest gain, and to 24.9\% for
max-length cycles, a nearly fourfold reduction. On the tested
instances, $\lambda_{geo}$ is therefore the more accurate estimate
under log-normal weights; whether this holds for other skewed
distributions and graph families remains to be tested.

\section{Limitations}
\label{sec:limits}

While our experimental results demonstrate the efficacy of the
proposed approach, we acknowledge several limitations in the current
study, ranging from empirical constraints to theoretical boundaries.

\paragraph{Sensitivity to Weight Distributions.}
As discussed in Section~\ref{sec:benchmark-results}, the performance
of our heuristic estimates exhibits a dependency on the underlying
edge weight distribution. While the method achieves accurate estimates
on symmetric (uniform) distributions, we observed degradation when
applied to graphs with skewed (log-normal) weight distributions. This
suggests that the heuristic is less effective when edge weight
variance is high.

\paragraph{Generalization to Other Domains.}
Our evaluation used ISCAS benchmark circuits, real circuit topologies
with synthetically assigned edge weights. We restricted our evaluation
to a subset of instances because exact cycle enumeration did not
terminate on circuits that were either larger or contained too many
cycles. While ISCAS circuits are representative of digital hardware,
the structural properties of graphs in other domains (e.g., biological
pathways, social networks) may differ, potentially affecting the size
of the critical subgraph and hence the tightness of our bounds.

\paragraph{Scalability and Hardware Constraints.}
Our experiments were constrained by the runtime of exhaustive cycle
enumeration, which was required to establish ground truth. As a
result, we restricted our analysis to graphs with at most one million
cycles. The optimum cycle means themselves were computed in under one
second for all benchmarks; practical applications of our bounds would
not require enumeration and thus would scale to much larger instances.

\paragraph{Complexity of Optimal Critical Cycles.}
Regarding the theoretical application of our bounds, determining the
``best'' critical cycles (e.g., the longest possible max-critical
cycle) is itself NP-hard. Nonetheless, since critical cycles lie
within the critical subgraph, restricting search to this subgraph can
substantially reduce the problem size. When the critical subgraph is
small enough, exhaustive cycle enumeration \citep{Jo75,HaJa08,BaSa25}
can become feasible.

\paragraph{Triviality on Uniform Weights.}
Finally, our bounds are trivial for uniformly weighted or unweighted
graphs: if all edges have weight 1, then $\lambda_{min} =
\lambda_{max} = 1$.

\section{Applications}
\label{sec:applications}

Despite the limitations regarding sensitivity to weight distributions
discussed in Section~\ref{sec:limits}, our proposed bounds remain
highly effective for a wide range of practical domains where edge
weights are bounded and the graph topology supports short critical
cycles. In this section, we highlight three specific applications
where these conditions are satisfied.

\paragraph{Application to Branch-and-Bound.}
A critical component of Branch-and-Bound algorithms for the longest
simple path or cycle is the \textit{admissible heuristic}: an
optimistic estimate of the remaining path extension. Our bounds
naturally serve as admissible heuristics in such frameworks
\citep{MoJaSaSe16,StPuFe14,CoStFe20}.  For example, when searching for
a cycle of length at most $k$ through a start node $s$, given a
partial path of $p$ edges ending at node $t$, at most $k - p$ edges
remain to complete the cycle. The remaining edges form a path rather
than a cycle, so $\lambda_{max}$ alone does not bound their weight: in
a two-edge cycle with weights $1$ and $99$, $\lambda_{max} = 50$, yet
after traversing the weight-$1$ edge the remaining edge weighs
$99$. Instead, we use node potentials $d_{max}$, the $\lambda_{max}$
counterpart of the potentials $d$ in the linear programming
formulation of Section~\ref{sec:opt-means}, obtained by applying that
formulation to the negated weights. They satisfy $w(u,v) \leq
\lambda_{max} + d_{max}(u) - d_{max}(v)$ for every edge
$(u,v)$. Summing along any completion path from $t$ to $s$ with $q
\leq k - p$ edges, the potentials telescope and the completion weighs
at most $q\lambda_{max} + d_{max}(t) - d_{max}(s)$, which is at most
$(k - p)\lambda_{max} + d_{max}(t) - d_{max}(s)$ when $\lambda_{max}
\geq 0$. This
estimate is admissible, enabling the pruning of branches that cannot
exceed the current best solution, and it costs nothing beyond the
single optimum cycle mean computation, which produces the potentials
as a by-product. While it may loosen under heavy-tailed weight
distributions, it never underestimates.

\paragraph{Application to Longest Simple Paths.}
Our cycle bounds induce bounds on longest simple paths. Since any
cycle can be converted to a path by removing one edge, the properties
of the longest cycles ($L_w, L_l$) impose strict lower bounds on the
longest paths:

\begin{cor}[Path bounds]
Let $P_w$ be a max-weight simple path and $P_l$ be a max-length simple
path, allowing a single node as a path of length and weight zero. Let
$\mathcal{W}$ be the set of critical cycles found and $B_l$ the
critical witness bound. Then
\begin{align}
    |P_l| &\geq |L_l| - 1 \geq B_l - 1, \\
    w(P_w) &\geq \max\limits_{C \in \mathcal{W}}
      \left( w(C) - \min\limits_{e \in C} w(e) \right)
      \geq \max\limits_{C \in \mathcal{W}} \left( (|C| - 1)\lambda(C) \right).
\end{align}
\end{cor}

\begin{proof}
Removing one edge from a simple cycle yields a simple path with one
fewer edge. Applying this to $L_l$ and using $|L_l| \geq B_l$ gives
the length bound. For any cycle $C$, removing its lightest edge gives
a path of weight $w(C) - \min\limits_{e \in C} w(e)$, and the lightest edge
weighs at most the average $\lambda(C)$; since this holds for every $C
\in \mathcal{W}$, it holds for the maximum over $\mathcal{W}$.
\end{proof}

In particular, any max-critical cycle found gives
\begin{equation}
    w(P_w) \geq (\max\limits_{C \in \mathcal{C}_{max}} |C| - 1) \lambda_{max},
\end{equation}
which is non-negative whenever $\lambda_{max} \geq 0$. Because
$w_{max} \geq \lambda_{max} \geq \lambda_{min}$, it dominates the
bound $(\max |C|)\lambda_{min} - w_{max}$ obtained by removing an
edge of unknown weight from $L_w$. Scanning the edges of each witness
gives the
stronger first bound of the corollary, obtained by removing the
actual lightest edge.

This corollary is particularly useful for initialization in
Branch-and-Bound algorithms \citep{MoJaSaSe16,StPuFe14,CoStFe20}. By
identifying a long or heavy critical cycle in polynomial time, one can
establish a high-quality lower bound (incumbent solution) for the
longest path problem before the search begins.

However, we note that these path bounds inherit the limitations
discussed in Section~\ref{sec:limits}. Critical cycles are typically
short (typically length 3--13 in our ISCAS experiments, with a single
witness of length 26), while longest paths can
approach $n-1$ edges, making the length bound loose in sparse
graphs. Likewise, the weight bound on $w(P_w)$ grows only linearly
with the witness length, so it can fall far short of the true value
when critical cycles are short. Thus, these bounds are most powerful
when the graph contains long critical cycles.

\paragraph{Application to Graph Structural Bounds.}
Literature on cycle length bounds is extensive
(\S~\ref{sec:related-work}). Given a structural bound that guarantees
a cycle $C$ with exactly $\ell$ edges, our framework automatically
translates it into the weight bound $w(L_w) \geq w(C) \geq \ell
\lambda_{min}$, valid for either sign of $\lambda_{min}$. A guarantee
of a cycle with at least $\ell$ edges gives $|L_l| \geq \ell$ and, when
$\lambda_{min} \geq 0$, $w(L_l) \geq \ell \lambda_{min}$. This
observation suggests that any advance in bounding cycle lengths for
specific classes of digraphs can be directly leveraged to bound cycle
weights via optimum cycle means.

\section{Related Work}
\label{sec:related-work}

\paragraph{Optimum Cycle Means.}
\citet{Da04} surveys $O(mn)$-time algorithms for computing optimum
cycle means and presents their experimental analysis results. The
efficient implementations of most of those algorithms are available in
\citep{Da15}.

\paragraph{Longest Paths and Cycles.}
The Longest Path/Cycle problem is a classical NP-hard problem
\citep{Ga79} and inapproximable within $n^{1-\epsilon}$
\citep{BjHuKh04}. Algorithmic approaches include color-coding
\citep{AlYuZw94}, which finds paths of length $O(\log n)$ in
polynomial time, and improvements by \citet{BjHu03}. For the problem
of finding paths or cycles of a specific length $k$, \citet{Mo85}
proposes a fixed-parameter tractable algorithm that runs in
$O(k!\,nm)$ time. As reported by \citet{HaLeCo24}, the algorithm of
\citet{Bo93} finds a path of length at least $k$ in $O(2^k k!\,n)$
time and an $s$-$t$ path of length at least $k$ in $O(2^{2k}(2k)!\,n +
m)$ time.

\paragraph{Structural Bounds.}
Literature on cycle length guarantees is extensive, e.g.,
\citet{ChZh88,Gh60,Th81} for directed graphs and
\citep{ErGa59,Di52,JaKrSo19,An22,FoGoSa24b} for undirected
graphs. While these works focus strictly on length, our algebraic
framework bridges these structural guarantees to weighted bounds, as
discussed in \S~\ref{sec:applications}.

\paragraph{Cycle Enumeration.}
\citet{Jo75} gives an output-sensitive algorithm for enumerating all
simple cycles; \citet{HaJa08} extend it to multigraphs. \citet{GuSu21}
address bounded-length enumeration.

\paragraph{Heuristic Search.}
\citet{MoJaSaSe16,StPuFe14,CoStFe20} develop A* and branch-and-bound
methods for longest simple paths using admissible heuristics. Our
bounds provide such heuristics via $\lambda_{max}$.

\section{Conclusions}
\label{sec:conclusions}

We have presented algebraic bounds for the weight and length of the
longest simple cycle based on optimum cycle means. While these bounds
do not yield constant-factor approximations, they are computable in
strongly polynomial time and provide necessary conditions for cycle
existence. Our experimental analysis on the ISCAS benchmarks yields
two practical conclusions:

First, the strict algebraic lower bounds are often loose (median
errors of 57--95\%) regardless of the weight distribution, as critical
cycles ($C_{min}$, $C_{max}$) are typically short while longest cycles
($L_w$, $L_l$) are long. However, the quality of \textit{heuristic estimates}
depends heavily on the distribution of edge weights. For symmetric
uniform distributions, the arithmetic mean $\lambda_{avg}$ is an
accurate estimate, achieving median errors of 12.7\% ($L_w$) and
5.8\% ($L_l$). For skewed log-normal distributions (typical of circuit
delays), $\lambda_{avg}$ degrades significantly, but the geometric
mean $\lambda_{geo}$ provides a substantial improvement---reducing
median errors from 91.9\% to 24.9\% for max-length cycles.

Second, we observed that max-weight and max-length cycles are
structurally similar and often identical, particularly in the uniform
case. This suggests that algorithms designed to find long paths are
likely to find heavy paths in practice, and vice versa. Future work
lies in integrating these bounds into a live Branch-and-Bound solver
to empirically measure the reduction in search nodes and total runtime
for large-scale instances.

\bibliographystyle{plainnat}
\bibliography{graph}

\end{document}